\documentclass[letterpaper, 10 pt, conference]{ieeeconf}

\usepackage{microtype}
\usepackage{graphicx}
\usepackage{booktabs} 
\usepackage{hyperref}

\IEEEoverridecommandlockouts                              
\usepackage{amsmath}
\usepackage{amssymb}
\usepackage{mathtools}
\usepackage{amsthm}
\usepackage{comment,nicefrac}
\usepackage{algorithm}
\usepackage{algorithmic}
\usepackage[font=small,labelfont=bf]{caption}
\usepackage{subcaption}
\usepackage{xcolor}
\usepackage[outdir=./]{epstopdf}
\usepackage{listings}
\usepackage{soul}

\newtheorem{theorem}{Theorem}[section]

\newtheorem{lemma}[theorem]{Lemma}

\theoremstyle{definition}
\newtheorem{definition}[theorem]{Definition}
\newtheorem{assumption}[theorem]{Assumption}
\theoremstyle{remark}

\newcommand{\vertiii}[1]{{\left\vert\kern-0.25ex\left\vert\kern-0.25ex\left\vert #1 
    \right\vert\kern-0.25ex\right\vert\kern-0.25ex\right\vert}}

\newcommand{\EE}{\mathbb{E}}

\newcommand{\Ss}{\mathcal{S}}
\newcommand{\Ls}{\mathcal{L}}

\newcommand{\Ns}{\mathcal{N}}
\newcommand{\Ps}{\mathcal{P}}

\newcommand{\Os}{\mathcal{O}}

\newcommand{\RR}{\mathbb{R}}

\newcommand{\uJ}{\ddot{J}}
\newcommand{\uK}{\ddot{K}}
\newcommand{\uS}{\ddot{S}}

\newcommand{\uq}{\ddot{q}}
\newcommand{\uP}{\ddot{P}}

\newcommand{\uSigma}{\ddot{\Sigma}}

\newcommand{\utau}{\ddot{\tau}}
\newcommand{\ugamma}{\ddot{\gamma}}

\newcommand{\tP}{\tilde{P}}

\newcommand{\tq}{\tilde{q}}

\newcommand{\tr}{Tr}

\definecolor{NavyBlue}{rgb}{0.0, 0.0, 0.5}
\definecolor{ceruleanblue}{rgb}{0.16, 0.32, 0.75}

\definecolor{green1}{rgb}{0.2,0.7,0.2}

\definecolor{CadetBlue}{rgb}{0.37, 0.62, 0.63}

\begin{document}
\title{Policy Optimization finds Nash Equilibrium in Regularized\\ General-Sum LQ Games
}

\author{Muhammad~Aneeq~uz~Zaman, Shubham~Aggarwal, Melih~Bastopcu, and Tamer~Ba{\c s}ar
\thanks{Research of the authors was supported in part by the US Army Research Office (ARO) Grant W911NF-24-1-0085.
}
%
\thanks{
The authors are with the Coordinated Science Laboratory at the University of Illinois Urbana-Champaign, Urbana, IL, USA-61801.
(Emails: \texttt{\{mazaman2,sa57,bastopcu,basar1\}@illinois.edu})
}
}


\maketitle

\begin{abstract}
In this paper, we investigate the impact of introducing relative entropy regularization on the Nash Equilibria (NE) of General-Sum $N$-agent games, revealing the fact that the NE of such games conform to linear Gaussian policies. Moreover, it delineates sufficient conditions, contingent upon the adequacy of entropy regularization, for the uniqueness of the NE within the game. As Policy Optimization serves as a foundational approach for Reinforcement Learning (RL) techniques aimed at finding the NE, in this work we prove the linear convergence of a policy optimization algorithm which (subject to the adequacy of entropy regularization) is capable of provably attaining the NE. Furthermore, in scenarios where the entropy regularization proves insufficient, we present a $\delta$-augmentation technique, which facilitates the achievement of an $\epsilon$-NE within the game.
\end{abstract}
\section{Introduction}
Recent years have witnessed a significant increase in the number of empirical works on reinforcement learning (RL) focusing on multi-agent competitive environments, in applications such as 
autonomous driving \cite{shalev2016safe}, multi-armed bandits \cite{de2018comparing}, healthcare \cite{yu2021reinforcement}, and finance \cite{charpentier2021reinforcement} among others. In many of these real-world problems \cite{yu2021reinforcement,charpentier2021reinforcement}, guaranteed convergence and sample efficiency are critical due to the high cost of incorrect decisions. Hence, considerable theoretical attention has been directed towards \emph{specific} instances of games, such as (1) purely cooperative games \cite{min2023cooperative}, (2) zero-sum games \cite{jin2021v}, (3) potential games \cite{song2021can,cheng2024provable}, or (4) (coarse) correlated equilibria in general-sum games \cite{zhang2022optimal,mao2023provably}. However, learning the Nash equilibrium (NE) for general-sum games is considered to be a significantly harder problem \cite{jin2021v}. Recently, a burgeoning area of research has emerged \cite{zhang2021multi}, focusing on the development of provable Nash finding using RL techniques tailored for large population games through the mean-field game (MFG) paradigm \cite{cui2021approximately,angiuli2022unified,zaman2023oracle}, but they only provide approximate solutions for finite agent scenarios. 

The quest for provable RL methods applicable to general-sum games has remained a formidable challenge. Notably, investigations have unveiled discouraging findings, demonstrating the inadequacy of policy gradient methods to converge to Nash equilibria even for the benchmark Linear-Quadratic (LQ) setting \cite{mazumdar2019policy}. Building upon this, the authors in \cite{hambly2023policy} have posited the significance of noise as a potential catalyst for enabling the analysis of policy gradient methods. Yet, the System Noise (Assumption 4 \cite{hambly2023policy}) condition articulated in their work finds a more generalized counterpart in the form of a diagonal dominance condition \cite{zaman2024independent}, which may not be satisfied solely in the presence of adequate noise in the system.

In this work we investigate Policy Optimization (PO) methods to find NE of General-Sum LQ games with \emph{relative entropy regularized} costs, a framework previously explored by \cite{cui2021approximately} for infinite agent games. PO methods serve as fundamental components, laying the groundwork for the development of data-driven RL methodologies such as policy gradient \cite{fazel2018global} and stochastic optimization \cite{mei2021understanding}. Our investigation reveals that the inclusion of entropy regularization within the cost function confines NE policies to an exclusively linear Gaussian domain. Moreover, introducing Gaussian noise into the policies facilitates exploration across the policy space, thereby contributing to the linear convergence of PO methods, particularly when the entropy regularization is sufficiently large.  Subsequently, we draw upon pertinent literature within the domain to contextualize our findings.

\textbf{Literature Review}
Policy optimization algorithms for the LQ regulator case were first proposed in \cite{fazel2018global} with guarantees on their global convergence. For the multi-agent game setting, policy gradient methods for zero-sum games have been considered in \cite{bu2019global,zhang2019policy}. Specifically, the paper \cite{bu2019global} proposes projection-free sequential algorithms with global sublinear convergence to the NE, and the paper \cite{zhang2019policy} proposes nested-gradient based algorithms with globally sublinear and locally linear convergence rates. As for general sum game settings, \cite{roudneshin2020reinforcement} establishes the uniqueness of NE and consequent global convergence of the policy gradient method for symmetric games (where the agents have the same dynamics and cost function) with an infinite horizon setting. Finally, the paper \cite{hambly2023policy} proposes a natural policy gradient method guaranteeing global linear convergence for finite horizon general-sum games, under a System Noise condition. For this setting \cite{zaman2024independent} proves linear convergence of stochastic gradient methods under a diagonal dominance condition, which is shown to be more general than the System Noise condition. Concurrent work of \cite{lidard2024blending} provides a similar NE policy characterization without theoretical guarantees on the uniqueness \& computation of the same.

To aid in the convergence of aforementioned algorithms and to explore the uncharted portions of the state space, entropy regularization has been considered as a widely sought-after approach \cite{wang2018exploration,wang2020continuous,guo2023fast}, in the single-agent setting. In particular, the authors in \cite{wang2018exploration} study the effect of using a differential-entropy based regularization in the objective function of a LQ regulator 
problem. It is observed that the optimal feedback policy has a Gaussian distribution and that the exploitation and exploration are (mutually-exclusively) captured by its mean and variance. The same framework is then extended to the study of the mean-variance portfolio selection problem in a finite horizon setting \cite{wang2020continuous}, which is essentially a LQ-type problem. The authors in \cite{guo2023fast} propose policy learning algorithms for a class of discounted LQ problems in an infinite horizon setting, under an additive Shannon entropy-type regularization in the objective function. Finally, it is also worth mentioning that the idea of entropy regularization has also been studied for stability of RL algorithms for MFGs \cite{guo2022entropy,firoozi2022exploratory}, where the framework essentially utilizes an (auxiliary) infinite agent game to entail a tractable characterization of approximate equilibrium solutions for the $N$--agent game problem \cite{lasry2007mean,huang2006large,aggarwal2022linear}.

In contrast to the previous works which focus on a single agent LQ regulator problem or on MFGs, the current work addresses the convergence of policy optimization in \emph{general-sum LQ games}, which is a considerably more challenging problem. We summarize the main contributions of our work as follows.

\textbf{Contributions: } 
\begin{enumerate}
    \item We formalize the Relative Entropy-Regularized General-Sum (ERGS) $N$-agent games where the dynamics of the state of the system is linear and the costs of the agents are quadratic.
    \item We characterize the NE of the ERGS game and find them to exclusively belong to the class of Linear Gaussian policies defined by a sequence of Coupled Riccati equations (Theorem \ref{thm:NE}). Furthermore, we find that if the entropy regularization parameter $\tau$ is lower bounded by a model dependent scalar, then the uniqueness of the NE is guaranteed.
    \item We propose a PO algorithm (Algorithm \ref{alg:RL_ER_Games}) which utilizes a receding-horizon approach to find the NE of the ERGS game. Under a suitable condition on $\tau$, we prove that the PO algorithm converges to the NE of the ERGS (Theorem \ref{thm:main_res}). Furthermore, in the case where $\tau$ does not satisfy the aforementioned condition, we provide a $\delta$-augmentation technique which yields an $\epsilon$-NE for the ERGS game.
\end{enumerate}

\textbf{Notation: } We denote the trace of a matrix as $\tr(\cdot)$, its determinant as $|\cdot|$, and its Frobenius norm as $\|\cdot\|_F$. A Gaussian distribution with mean $\mu$ and covariance $\Sigma$ is denoted by $\Ns(\mu,\Sigma)$. A policy $\mu$ is standard normal if it is distributed as $\Ns(0,I)$. The notation $[M]$ signifies the set $\{1,\ldots,M\}$ for all $M \in \mathbb{N}$ and $I$ denotes an indentity matrix of appropriate dimensions. 
Further, we use the notation $Z \succeq 0$ ($Z \succ 0$) for a symmetric positive semidefinite (definite) matrix $Z$. The notation $\mathcal{P}(\mathcal{X})$ denotes the set of probability distributions over $\mathcal{X}$. Finally, with a possible abuse of notation, we use the symbol log for natural logarithm.


\section{Problem Formulation}
Consider a general-sum $N$--agent game, where the joint state dynamics $x_t \in \RR^m$ is affected by the control actions $u^i_t \in \RR^p$ of each agent $i \in [N]$, and follows the stochastic difference equation:
\begin{align}
\begin{split}
x_{t+1} = A_t x_t + \sum_{i = 1}^N B^i_t u^i_t + \omega_t, \quad 0 \leq t \leq T-1, \label{eq:dyn}
\end{split}
\end{align}
where $x_0 \sim \Ns(x^0,\Sigma^0)$, $\Sigma^0$ is a positive semi-definite symmetric matrix, and $A_t,$ $B^i_t$ are matrices of suitable dimensions. Further, the noise process is $\omega_t \sim \Ns(0,\Sigma)$ (where $\Sigma \succeq 0$), which and is independent of the state $x_t$ for all $t$ as well as the initial state distribution.

Each agent $i \in [N]$ aims to minimize its finite horizon cost function given by: 
\begin{align}\nonumber\\[-2.2em]
     J^i (x,\pi^i,\pi^{-i}) =& \EE \bigg[ \sum_{t=0}^{T-1} x^\top_t Q^i_t x_t + (u^i_t)^\top R^i_t u^i_t  \label{eq:ent_reg_cost} \\
    & \hspace{-0.1cm} + \tau \log \bigg( \frac{\pi^i_t(u^i_t | x_t)}{\mu^i_t(u^i_t | x_t)} \bigg) + x^\top_T Q^i_T x_T \bigg| x_0 = x \bigg], \nonumber
\end{align}
 which depends on its own policy $\pi^{i}$ as well as the policies of other agents $\pi^{-i}$ (through the state). 
 The matrices $Q^i_t \succeq 0$ and $R^i_t \succ 0$ are symmetric matrices and the logarithmic term is the relative entropy-based regularizer (as in \cite{cui2021approximately}), which equivalently denotes the KL divergence between a policy $\pi^i$ and a prior policy $\mu^i$. The parameter $\tau > 0$ weights the amount of entropy added in the cost function. The policies $(\pi^i)_{i \in [N]}$ lie in the space of admissible policies $\Pi^i:= \{\pi^i \mid \pi^i_t \text{ is adapted to } \mathcal{I}^i_t:= \{x_0, u^i_0, x_1,\ldots,u^i_{t-1}, x_t\}\}$, where $\mathcal{I}^i_t$ denotes the information available to agent $i$ at time $t$.

For each agent $i \in [N]$ to minimize its cost \eqref{eq:ent_reg_cost} with respect to their policies $\pi^i$, the appropriate solution concept is the NE which is defined below.
\begin{definition} \label{def:NE}
    A NE of an $N$--agent game is a set of policies $(\pi^{1*},\ldots,\pi^{N*})$ such that for any $x \in \RR^m$,
    \begin{align*}
        J^i(x,\pi^{i*},\pi^{-i*}) \leq J^i(x,\pi^{i},\pi^{-i*})
    \end{align*}
    for all $i \in [N]$ and $\pi^i \in \Pi^i$.
\end{definition}
Next we characterize the NE of the ERGS game.
\section{Nash Equilibrium Characterization}
In this section we first characterize the NE of the ERGS game using the Hamilton-Jacobi-Isaacs 
 (HJI) equations. Then we provide sufficient conditions for uniqueness of the NE.

To aid us in the above characterization, we start by defining the value function of agent $i\in [N]$ at time $t$ as a partial cost-to-go from time $t$ to $T$:
\begin{align}
     J^i_t (x,\pi^i,\pi^{-i}) =& \EE \bigg[ \sum_{s=t}^{T-1} x^\top_{s} Q^i_{s} x_{s} + (u^i_{s})^\top R^i_{s} u^i_{s}  \label{eq:ent_reg_value} \\
    &  + \tau \log \bigg( \frac{\pi^i_s(u^i_{s} | x_{s})}{\mu^i_{s}(u^i_{s} | x_{s})} \bigg) + x^\top_T Q^i_T x_T \bigg| x_t \!=\! x \bigg]. \nonumber
\end{align}
Further, we define the Nash value of the same agent $i$ as the cost incurred under a NE policy $J^{i*}_t(x) = J^i_t(x,\pi^{i*},\pi^{-i*})$. Before characterizing the NE of the ERGS game we provide a helper lemma which might be of independent interest. 
    \begin{lemma} \label{lem:helper}
        Suppose that $M \in \RR^{p \times p}$ is a positive semi-definite symmetric matrix, $\tau > 0$, $b \in \RR^p$ and prior policy $\mu(u) = \Ns(0,I)$. Then, the probability distribution $\pi(u) \in \Ps(\RR^p)$ which minimizes the following expression,
        \begin{align}
 \EE_{u \sim \pi(\cdot)} \bigg[u^\top M u + b^\top u + \tau \log \bigg( \frac{\pi(u)}{\mu(u)} \bigg) \bigg], \label{eq:helper_cost}
        \end{align}
        is a multivariate Gaussian distribution, in particular,  $\pi(u) = \Ns(-((\nicefrac{\tau}{2})I + M)^{-1} b/2, (I + 2M/\tau)^{-1})$. Furthermore, the cost \eqref{eq:helper_cost} is strictly convex in $\pi$.
    \end{lemma}
    \begin{proof}
        Let us first reformulate the problem into a constrained optimization problem as:
        \begin{align}
             \min_{\pi(\cdot)}& \quad  \EE_{u \sim \pi(\cdot)} \bigg[u^\top M u + b^\top u + \tau \log \bigg( \frac{\pi(u)}{\mu(u)} \bigg) \bigg] \nonumber \\
             s.t. & \quad \int_{\RR^p} \pi(u) du = 1, \hspace{0.2cm} \pi(u) > 0. \label{eq:lag_const}
        \end{align}
        The constraint $\int_{\RR^p} \pi(u) du = 1$ is then incorporated using a Lagrange multiplier (non-negativity is shown to be satisfied implicitly) in the Lagrangian, defined by
        \begin{align*}
         \Ls(\pi,\lambda) =& \int_{\RR^p} \bigg( u^\top M u + b^\top u + \tau \log \bigg( \frac{\pi(u)}{\mu(u)} \bigg) \bigg) \pi(u) du \\
            &  + \lambda \bigg(  \int_{\RR^p} \pi(u) du - 1 \bigg),
            \end{align*}
which can then be equivalently written as:
            \begin{align*}
            \Ls(\pi,\lambda)& = \!\int_{\RR^p} \!\!\!\bigg( \!u^\top M u \!+\! b^\top u \!+\! \tau \log\! \bigg( \!\frac{\pi(u)}{\mu(u)} \bigg) \!+\! \lambda\!\bigg) \pi(u) du \!-\! \lambda \\
            & = \int_{\RR^p} f(u, \pi, \mu, \tau,\lambda) du - \lambda,
        \end{align*}
        where $f(u, v, w, \tau, \lambda) = (u^\top M u + b^\top u + \tau \log(v/w) + \lambda) v$. A necessary condition for optimality is $\partial f / \partial v = 0$, which leads to
        \begin{align}
            \frac{\partial f}{\partial v} = u^\top M u + b^\top u + \tau \log \bigg(\frac{v}{w} \bigg) + \tau + \lambda = 0. \label{eq:help_der}
        \end{align}
        Using \eqref{eq:help_der}, we have that
        \begin{align*}
            v = \frac{w}{e^{1 + \frac{\lambda}{\tau}}} \exp \bigg(-\frac{1}{\tau} (u^\top M u + b^\top u)\bigg). 
        \end{align*}
        Further, since $\frac{\partial^2}{\partial v^2} \big( u^\top M u + b^\top u + \tau \log \big( \frac{v}{w} \big) \big) v  = \frac{\tau}{v} > 0,$
        the cost \eqref{eq:helper_cost} is strictly convex with respect to $\pi$. An implicit assumption here is that $v$ here cannot be zero, which means that $\pi(u)$ assigns positive probability to every open subset of the Euclidean space where $u$ belongs. Then, using the definition of $\mu(u) = \Ns(0,I)$, we get that
        \begin{align}
            \hspace{-0.2cm} \pi(u) \!=\! \frac{1}{(2\pi)^{\frac{p}{2}}e^{1 + \frac{\lambda}{\tau}}} \!\exp\! \bigg( \!\!\!- \frac{1}{2} u^\top\! u - \frac{1}{\tau} (u^\top\! M u + b^\top u)\bigg).\!\!\! \label{eq:pi_form}
        \end{align}
        By simplifying the exponent in \eqref{eq:pi_form} by completion of squares, we obtain
        \begin{align}\nonumber\\[-2.2em]
            & -\frac{1}{2} u^\top u - \frac{1}{\tau} u^\top M u - \frac{1}{\tau} b^\top u \nonumber  \\
            & = - \frac{1}{2} \bigg(u + 2\bigg( \frac{\tau}{2}I +  M \bigg)^{-1} \!\!
            b\bigg)^\top \bigg( I + \frac{2}{\tau} M \bigg) \nonumber \\
            & \hspace{0.9cm} \bigg(u + 2\bigg( \frac{\tau}{2}I +  M \bigg)^{-1} \!\!b\bigg) + \frac{b^\top (I + 2M/\tau)^{-1}b}{2 \tau^2}, \!\!\label{eq:exp_decomp}
        \end{align}
where the existence of $((\nicefrac{\tau}{2}) I + M)^{-1}$ follows since $M \succeq 0$. Consequently, by using \eqref{eq:pi_form} and \eqref{eq:exp_decomp}, $\pi(u)$ can be found as 
        \begin{align*}\nonumber\\[-2.2em]
            \pi(u) \!=\! c(\lambda) \cdot e^{- \frac{1}{2} \big(u + 2\big( \frac{\tau}{2}I +  M \big)^{-1} \!\!b\big)^\top \big( I + \frac{2}{\tau} M \big)\big(u + 2\big( \frac{\tau}{2}I +  M \big)^{-1} \!\!b\big) }\\[-2.2em]
        \end{align*}
        where $c(\lambda)$ is a function of $\lambda$. Since $M \succeq 0$ and $\tau > 0$, the inverse term in the above expression of $\pi(u)$ exists. As 
        $\pi(u)$ has the form of Gaussian distribution we can choose the Lagrange multiplier $\lambda$ to satisfy the constraint in \eqref{eq:lag_const}. Hence, the distribution $\pi$ which solves the optimization problem in \eqref{eq:lag_const} is given by,
        \begin{align*}\\[-2.2em]
            \pi(u) = \Ns \big(-2\big( (\nicefrac{\tau}{2})I +  M \big)^{-1} b, \big( I + \nicefrac{2}{\tau} M \big)^{-1} \big).\\[-2.2em]
        \end{align*}
        The proof of the lemma is thus complete.
    \end{proof}
This lemma helps us characterize the NE of the ERGS game \eqref{eq:dyn}-\eqref{eq:ent_reg_cost} presented in the following theorem.
\begin{theorem} \label{thm:NE}
    Suppose that the prior policies $\mu^i$ are standard normal 
    for all $i$. Then, the Nash value of the $i^{th}$ agent is quadratic, i.e., $J^{i*}_t(x) = x^\top P^{i*}_t x + q^{i*}_t$ with $P^{i*}_t \succeq 0$ satisfying  the coupled Riccati equations
    \begin{align}
        & P^{i*}_{t} = Q^i_t + (A^{i*}_t)^\top P^{i*}_{t+1} A^{i*}_t - ((B^i_t)^\top P^{i*}_{t+1} A^{i*}_t)^\top \label{eq:CARE}\\
        & \hspace{0.1cm}  \big((\nicefrac{\tau}{2}) I + R^i_t + (B^i_t)^\top P^{i*}_{t+1} B^i_t \big)^{-1} ((B^i_t)^\top P^{i*}_{t+1} A^{i*}_t), \nonumber 
    \end{align}
    where $P^{i*}_T = Q^i_T$, $A^{i*}_t := A_t + \sum_{j \neq i} B^j_t K^{j*}_t$ and $q^{i*}_t \in \RR$ is defined recursively as
    \begin{align}\nonumber\\[-2.2em]
        q^{i*}_t = &\tr(\Sigma P^{i*}_{t+1}) + \tr((R^i_t + (B^i_t)^\top P^{i*}_{t+1} B^i_t) \Sigma^{i*}_t \nonumber \\
        &  + \frac{\tau}{2} \left(\tr(\Sigma^{i*}_t) - p - \log(|\Sigma^{i*}_t|)\right) + q^{i*}_{t+1} \nonumber\\
        &  + \tr \bigg(\sum_{j \neq i}\Sigma^{j*}_t (B^j_t)^\top P^{i*}_{t+1} B^j_t \bigg), \label{eq:qARE} 
    \end{align}
    and $q^{i*}_T = 0$. In addition, the NE policies of the $N$--agent ERGS game are linear Gaussian, i.e., $u^{i*}_t = \pi^{i*}_t(x) \sim \Ns(K^{i*}_t x,\Sigma^{i*}_t)$ where,
    \begin{align}\nonumber\\[-2.2em]
        K^{i*}_t  =& - \bigg( (\nicefrac{\tau}{2})I + R^i_t + (B^i_t)^\top P^{i*}_{t+1} B^i_t  \bigg)^{-1} (B^i_t)^\top P^{i*}_{t+1} A^{i*}_t, \nonumber\\
        \Sigma^{i*}_t =& \big(I + 2(R^i_t + (B^i_t)^\top P^{i*}_{t+1} B^i_t)/\tau \big)^{-1}. \label{eq:Nash_pol}
    \end{align}
\end{theorem}
\begin{proof}
    We start by proposing that the form of Nash value function \eqref{eq:ent_reg_value} is quadratic, which we will prove by induction. Since we know that $J^{i*}_T(x) = x^\top Q^i_T x$, let
    \begin{align} \nonumber\\[-2.2em]
        J^{i*}_{t+1} = x^\top P^{i*}_{t+1} x + q^{i*}_{t+1}.
        \label{eq:induction}\\[-2.2em]\nonumber
    \end{align}
    for some $t \in [0,T-1]$. From HJI equations \cite{bacsar1998dynamic}, 
    we know 
    \begin{align*}\\[-2.2em]
        & J^{i*}_t (x) = \min_\pi \EE_{u \sim \pi(\cdot)} \bigg[ x^\top Q^i_t x + u^\top R^i_t u + \tau \log \bigg(\frac{\pi(u | x)}{\mu(u|x)} \bigg) \\
        & \hspace{3cm} + J^{i*}_{t+1} \big(L^{i*}_t + B^i_t u + \omega_t \big) \bigg],\\[-2.2em]
    \end{align*}
    where $L^{i*}_t(x) := A_t x + \sum_{j \neq i} B^j_t u^{j*}_t(x)$. By using \eqref{eq:induction}, we get
    \begin{align}
        & J^{i*}_t (x) = \min_\pi \EE_{u \sim \pi(\cdot)} \bigg[ x^\top Q^i_t x + u^\top R^i_t u + \tau \log \bigg(\frac{\pi(u | x)}{\mu(u|x)} \bigg) \nonumber \\
        & + q^{i*}_{t+1} + \big(L^{i*}_t(x) + B^i_t u + \omega_t \big)^\top \!P^{i*}_{t+1} \big(L^{i*}_t(x) + B^i_t u \!+ \!\omega_t \big)   \bigg] \nonumber \\
        & = x^\top Q^i_t x + \tr(\Sigma P^{i*}_{t+1}) \!+ \!q^{i*}_{t+1} + \EE \big[(L^{i*}_t(x))^\top P^{i*}_{t+1} L^{i*}_t(x) \big] \nonumber \\
        & \hspace{0.4cm} + \min_\pi \EE_{u \sim \pi(\cdot)} \bigg[ u^\top (R^i_t + (B^i_t)^\top P^{i*}_{t+1} B^i_t) u \nonumber \\
        & \hspace{1.5cm} +\! 2 u^\top\! (B^i_t)^\top P^{i*}_{t+1} L^{i*}_t(x) \!+\! \tau \log \bigg(\!\frac{\pi(u | x)}{\mu(u|x)}\! \bigg) \!\bigg]. \!\!\label{eq:value_inter}
    \end{align}
    Using Lemma \ref{lem:helper}, we deduce that the NE policy of agent $i$ is given as 
    \begin{align}
        & u^{i*}_t \sim \pi^{i*}_t(x) = \Ns \big( M^{i*}_t(x), \Sigma^{i*}_t\big), \label{eq:NE_pol_form}
    \end{align}
    where
    \begin{align}
        M^{i*}_t & = - \Big(\frac{\tau}{2}I + R^i_t + (B^i_t)^\top P^{i*}_{t+1} B^i_t \Big)^{-1} (B^i_t)^\top P^{i*}_{t+1} L^{i*}_t(x), \nonumber \\
        \Sigma^{i*}_t & = \big(I + 2(R^i_t + (B^i_t)^\top P^{i*}_{t+1} B^i_t )/\tau \big)^{-1}. \label{eq:NE_pol_part}
    \end{align}
    By direct calculation we can check that the set of policies 
        \begin{align}
        & u^{i*}_t \sim \pi^{i*}_t(x) = \Ns \big( K^{i*}_t x, \Sigma^{i*}_t\big), \text{ with} \label{eq:NE_pol_form_2} \\
        K^{i*}_t & = - \Big((\nicefrac{\tau}{2})I + R^i_t + (B^i_t)^\top P^{i*}_{t+1} B^i_t \Big)^{-1} (B^i_t)^\top P^{i*}_{t+1} A^{i*}_t, \nonumber \\
        \Sigma^{i*}_t & = \big(I + 2(R^i_t + (B^i_t)^\top P^{i*}_{t+1} B^i_t )/\tau \big)^{-1}, \label{eq:NE_pol_part_2}
    \end{align}
    satisfies the conditions \eqref{eq:NE_pol_form} and \eqref{eq:NE_pol_part}, and hence, also the necessary conditions for 
    the NE. Moreover, since the cost functional in \eqref{eq:value_inter} is strictly convex (due to Lemma \ref{lem:helper}), we see that
    \eqref{eq:NE_pol_form_2}-\eqref{eq:NE_pol_part_2} is also sufficient for 
 the NE. 
    
    Having proved the form of the NE policies, we now prove \eqref{eq:CARE} and \eqref{eq:qARE}.  
    Using results from \cite{robert1996intrinsic}, the negative relative entropy term can be computed as
    \begin{align}
        & \EE_{\pi^{i*}_t} \bigg[ \log \bigg(\frac{\pi^{i*}_t(u | x)}{\mu^{i}_t(u | x)} \bigg) \bigg] = \label{eq:neg_int}\\
        & \hspace{1.5cm} \frac{1}{2} \big( (K^{i*}_t x)^\top K^{i*}_t x + \tr(\Sigma^{i*}_t) - p - \log(|\Sigma^{i*}_t|) \big). \nonumber
    \end{align}
    Then, by substituting \eqref{eq:NE_pol_form_2} and \eqref{eq:neg_int} into \eqref{eq:value_inter}, we arrive at
    \begin{align*}
        & J^{i*}_t(x) = x^\top Q^i_t x + \tr(\Sigma P^{i*}_{t+1}) + (A^{i*}_t x)^\top P^{i*}_{t+1} A^{i*}_t x + q^{i*}_{t+1} \\
        & + \tr((R^i_t + (B^i_t)^\top P^{i*}_{t+1} B^i_t)\Sigma^{i*}_t) + 2 (B^i_t K^{i*}_t x)^\top P^{i*}_{t+1} A^{i*}_t x \\
        & + (K^{i*}_t x)^\top (R^i_t + (B^i_t)^\top P^{i*}_{t+1} B^i_t) K^{i*}_t x + \frac{\tau}{2}  \big( (K^{i*}_t x)^\top K^{i*}_t x \\
        & + \tr(\Sigma^{i*}_t) - p - \log(|\Sigma^{i*}_t|) \!+\! \tr \bigg(\sum_{j \neq i}\Sigma^{j*}_t (B^j_t)^\top P^{i*}_{t+1} B^j_t \bigg).
    \end{align*}
    Substituting \eqref{eq:NE_pol_part_2} in the above yields 
    \begin{align*}\\[-2em]
        & J^{i*}_t(x) = x^\top \bigg( Q^i_t  + (A^{i*}_t)^\top P^{i*}_{t+1} A^{i*}_t +  \\
        &  - 2 ( (B^i_t)^\top P^{i*}_{t+1} A^{i*}_t)^\top \Big(\frac{\tau}{2}I + R^i_t + (B^i_t)^\top P^{i*}_{t+1} B^i_t \Big)^{-1} \\
        & \hspace{6cm}(B^i_t)^\top P^{i*}_{t+1} A^{i*}_t \\
        &  - ( (B^i_t)^\top P^{i*}_{t+1} A^{i*}_t)^\top \Big(\frac{\tau}{2}I + R^i_t + (B^i_t)^\top P^{i*}_{t+1} B^i_t \Big)^{-2} \\
        & \hspace{6cm}(B^i_t)^\top P^{i*}_{t+1} A^{i*}_t \\
        &  - ( (B^i_t)^\top P^{i*}_{t+1} A^{i*}_t)^\top \Big(\frac{\tau}{2}I + R^i_t + (B^i_t)^\top P^{i*}_{t+1} B^i_t \Big)^{-1}\\
        & \hspace{0.9cm} (R^i_t + (B^i_t)^\top P^{i*}_{t+1} B^i_t) \Big(\frac{\tau}{2}I + R^i_t + (B^i_t)^\top P^{i*}_{t+1} B^i_t \Big)^{-1} \\
        & \hspace{6cm}(B^i_t)^\top P^{i*}_{t+1} A^{i*}_t \bigg)x \\
        & + \tr(\Sigma P^{i*}_{t+1}) + \tr((R^i_t + (B^i_t)^\top P^{i*}_{t+1} B^i_t)\Sigma^{i*}_t) \\
        & + \frac{\tau}{2} \big( \tr(\Sigma^{i*}_t) - p - \log(|\Sigma^{i*}_t|) \big) + q^{i*}_{t+1} \\
        & \hspace{3.5cm} + \tr \bigg(\sum_{j \neq i}\Sigma^{j*}_t (B^j_t)^\top P^{i*}_{t+1} B^j_t \bigg).
    \end{align*}
    Simplifying the above expression and using \eqref{eq:CARE}-\eqref{eq:qARE} finally yields
    \begin{align*}
        & J^{i*}_t(x) = x^\top \bigg( Q_t + (A^{i*}_t)^\top P^{i*}_t A^{i*}_t - (B^\top_t P^{i*}_{t+1} A^{i*}_t)^\top \end{align*}\begin{align*}
        & \hspace{0.1cm} \big(\nicefrac{\tau}{2} I + R^i_t + (B^i_t)^\top P^{i*}_{t+1} B^i_t \big)^{-1} (B^\top_t P^{i*}_{t+1} A^{i*}_t) \bigg) x \\
        & + \tr(\Sigma P^{i*}_{t+1}) + \tr((R^i_t + (B^i_t)^\top P^{i*}_{t+1} B^i_t)\Sigma^{i*}_t) \\
        & + \frac{\tau}{2} \big( \tr(\Sigma^{i*}_t) - p - \log(|\Sigma^{i*}_t|) \big) \\
        & + q^{i*}_{t+1} + \tr \bigg(\sum_{j \neq i}\Sigma^{j*}_t (B^j_t)^\top P^{i*}_{t+1} B^j_t \bigg)= x^\top P^{i*}_t x + q^{i*}_t,
    \end{align*}
    which completes the proof.
\end{proof}
The proof of Theorem \ref{thm:NE} uses HJI equations \cite{bacsar1998dynamic} which reduce the problem to finding the most optimal probability distribution for each agent (in each timestep) given that the other agents follow NE policies. This simplification coupled with the fact that entropy regularized quadratic function is minimized by a Linear Guassian probability distribution (Lemma \ref{lem:helper}), helps us arrive at the conclusion. 

Theorem \ref{thm:NE} also states that although there may be multiple NE (any solution to \eqref{eq:CARE} will specify a NE), all NE must have a linear Gaussian structure given by \eqref{eq:Nash_pol}. Since all policies belong to the class of linear Gaussian distributions, we will henceforth restrict our discussion to policies of the form $S^i_t = (K^i_t,\Sigma^i_t)$ with $K^i_t \in \RR^{p \times m}$ and $\Sigma^i_t \in \RR^{p \times p}$ 
where $\Sigma^i_t$ is a symmetric positive definite matrix, such that $\pi^i_t(x) = u^i_t \sim \Ns(K^i_t x, \Sigma^i_t)$. We also denote the joint policy at time $t \in \{0,\ldots,T-1 \}$ as $S_t := (S^i_t)_{ i \in [N]}$ and the joint policy for agent $i \in [N]$ as $S^i := (S^i_t)_{0 \leq t \le T-1}$. We further denote the \textit{set} of all linear Gaussian policies for agent $i$ as $\Ss^i_t$, joint policies at time $t$ by $\Ss_t$ and the set of all joint policies for agent $i$ by $\Ss^i$. In keeping with game theoretic notation, superscript $-i$ denotes the set of all policies except for agent $i$; in particular, we have that $\Ss^{-i} = (\Ss^j)_{ j \in [N]\setminus \{i\}}$.

To provide sufficient conditions for the existence and uniqueness of the Nash equilibrium, we first introduce a lower bound on the entropy regularization parameter $\tau$.
\begin{assumption} \label{asm:tau}
    Assume $\tau > 2 \gamma^2_B \gamma^*_P (N-1)$, where $\gamma_B := \max_{\forall i,t} \lVert B^i_t \rVert_F$ and $\gamma^*_P := \max_{\forall i,t} \lVert P^{i*}_t \rVert_F$.
\end{assumption}
This recursive assumption is similar in structure to the condition on attenuation parameter in robust control problems \cite{bacsar2008h,zaman2024robust}. Under this lower bound assumption on $\tau$, there exists a unique NE of the game \eqref{eq:dyn}-\eqref{eq:ent_reg_cost} as proved below.
\begin{lemma} \label{Lemma:III-4}
    Under Assumption \ref{asm:tau}, the NE of the ERGS LQ Game \eqref{eq:dyn}-\eqref{eq:ent_reg_cost} is unique.
\end{lemma}
\begin{proof}
    Using \eqref{eq:Nash_pol} for each $i \in [N]$, we have that
    \begin{align*}
        \Big(\frac{\tau}{2}I + R^i_t + (B^i_t)^\top P^{i*}_{t+1} B^i_t \Big)K^{i*}_t = - (B^i_t)^\top P^{i*}_{t+1} A^{i*}_t.
    \end{align*}
    This gives us a joint set of equations,
    \begin{align*}
        \Phi_t \begin{pmatrix} K^{i*}_t \\ \vdots \\ K^{N*}_t \end{pmatrix} = -\begin{pmatrix} (B^1_t)^\top P^{1*}_{t+1} A_t \\ \vdots \\ (B^N_t)^\top P^{N*}_{t+1} A_t \end{pmatrix},
    \end{align*}
    where $\Phi_t$ is a block diagonal matrix with 
    diagonal entries $(\nicefrac{\tau}{2}) I + R^i_t + (B^i_t)^\top P^{i*}_{t+1} B^i_t$ and the off-diagonal $(i,j)^{th}$ entries (for $i\neq j$) as $(B^i_t)^\top P^{i*}_{t+1} B^j_t$. Then, if $\tau > 2 \gamma^2_B \gamma^*_P (N-1)$ by Assumption \ref{asm:tau}, we have that $\Phi_t$ is diagonally dominant and hence, the linear components of the NE policies is uniquely determined by
    \begin{align*}\\[-2em]
        \begin{pmatrix} K^{i*}_t \\ \vdots \\ K^{N*}_t \end{pmatrix} = - \Phi^{-1}_t \begin{pmatrix} (B^1_t)^\top P^{1*}_{t+1} A_t \\ \vdots \\ (B^N_t)^\top P^{N*}_{t+1} A_t \end{pmatrix}.\\[-2em]
    \end{align*}
    The Riccati matrices $P^{i*}_t$ and the covariance matrices $\Sigma^{i*}_t$ are now determined by the uniquely solvable set of equations \eqref{eq:CARE}-\eqref{eq:Nash_pol}. The proof is thus complete.
\end{proof}
The proof of Lemma~\ref{Lemma:III-4} relies on the invertibility of the $\Phi_t$ matrix at each timestep $t$, similar to the invertibility conditions in non-regularized LQ games \cite{hambly2023policy}, but relies on the entropy regularization parameter. Having established sufficient conditions for the existence and uniqueness of NE \eqref{eq:Nash_pol}
, next we move to PO methods for computation of the same.

\section{Policy Optimization (PO) \& Non-Asymptotic Analysis}
In this section we first propose a PO algorithm to compute the NE of the ERGS \eqref{eq:dyn}-\eqref{eq:ent_reg_cost}. PO algorithms are important since they form the building blocks of data driven RL techniques. Then, under the assumption on the regularization parameter $\tau$, we prove linear convergence of the proposed algorithm to the NE of the game. Additionally, in Section \ref{subsec:delta_aug} we will also show that in the case that this assumption is not satisfied, we can still compute an $\epsilon$-NE of the game by artificially injecting regularization in the cost.
Consider a policy $S^i_t = (K^i_t, \Sigma^i_t) \in \Ss^i_t$ for agent $i \in [N]$ and $0 \leq t \leq T-1$, and a  joint policy over all agents $i \in [N]$ for a given $0 \leq t \leq T-1$ as $S_t:= (S^i_t)_{\forall i} = (K^i_t, \Sigma^i_t)_{\forall i}$. For the space of joint policies $\Ss_t = \prod_{i \in [N]} \Ss^i_t$, we define the metric
\begin{align*}
    \lVert S_t \rVert_\square = \sum_{i \in [N]} \lVert K^i_t \rVert_F + \lVert \Sigma^i_t \rVert_F.\\[-2em]
\end{align*}

Using Theorem \ref{thm:NE}, let us also define the best response map of agent $i$ at time $t$ to  the set of policies $S_t$ as $\Psi^i_t(\cdot) : \Ss_t \rightarrow \Ss^i_t$ such that
    \begin{align*}\\[-2em]
        \Psi^i_t \big(S_t|S_{[t+1:T-1]} \big) = (K^{i\prime}_t, \Sigma^{i\prime}_t), 
    \end{align*}
    where $S_{[s:t]} := \{S_{s},\ldots,S_{t}\}$ for $s \leq t$ and
    \begin{align}\nonumber\\[-2em]
        K^{i\prime}_t & = - \bigg( \frac{\tau}{2}I + R^i_t + (B^i_t)^\top P^i_{t+1} B^i_t  \bigg)^{-1} (B^i_t)^\top P^i_{t+1} A^i_t, \nonumber \\
        \Sigma^{i\prime}_t &= \big(I + 2(R^i_t + (B^i_t)^\top P^i_{t+1} B^i_t)/\tau \big)^{-1}. \label{eq:best_response}\\[-2em]\nonumber
    \end{align}
    and the $P^i_{t+1}$ matrices are obtained using Lyapunov equations 
    \begin{align*}
        & P^i_s = Q^i_s + (K^i_s)^\top ((\nicefrac{\tau}{2}) I +  R^i_s) K^i_s \\
        & \hspace{2cm} +  (A^i_s + B^i_s K^i_s)^\top P^i_{s+1} (A^i_s + B^i_s K^i_s),
    \end{align*}
    for $t+1 \leq s \leq T-1$ with $A^i_s = A_s + \sum_{j \neq i} B^j_s K^j_s$ and $P^i_T = Q^i_T$. The matrices $P^i_t$ are symmetric positive semidefinite by definition and consequently the inverses in \eqref{eq:best_response} exist. The composite best response operator is denoted by $\Psi_t(\cdot) : \Ss_t \rightarrow \Ss_t$ as defined as
    \begin{align*}
        \Psi_t(S_t|S_{[t+1:T-1]}) = \big(\Psi^i_t(S_t|S_{[t+1:T-1]})\big)_{\forall i}.
    \end{align*}
   From Definition \ref{def:NE}, we know that a joint policy $S^*_t$ which satisfies $S^*_t = \Psi_t\big(S^*_t|S^*_{[t+1:T-1]}\big)$ is also the NE of the ERGS LQ Game \eqref{eq:dyn}-\eqref{eq:ent_reg_cost}.

   In Algorithm \ref{alg:RL_ER_Games}, we adopt a receding-horizon methodology \cite{zhang2023revisiting}, wherein the NE policy at timestep $T-1$ is approximated (by multiple iterations of the best response operator $\Psi_t$), and followed by a backward-in-time progression to $T-2$ and subsequent time steps. Conceptually, this process mirrors an approximate solution to the HJI equations \cite{bacsar1998dynamic}. 
\begin{algorithm}[t!]
	\caption{PO for ERGS LQ Games}
	\begin{algorithmic}[1] \label{alg:RL_ER_Games}
		\STATE {Initialize $S^i_t = (K^i_t, \Sigma^i_t) =0$ for all $i \in [N], t \in \{0,\ldots,T-1\}$}
        \FOR {$t = T-1,\ldots,1,0,$}
        \FOR {$l = 1,\ldots,L$}
        \STATE $S^{i\prime}_t = \Psi^i_t \big( S_t | S_{[t+1:T-1]} \big)$, $\forall i \in [N]$
        \STATE $S^i_t \leftarrow S^{i\prime}_t$, $\forall i \in [N]$
        \ENDFOR
        \ENDFOR
        \STATE \textbf{Output: } $(S_t)_{0 \leq t \leq T-1}$
	\end{algorithmic}
\end{algorithm}

In the following lemma, we show that if $\tau$ satisfies a lower bound, then the $\Psi_t$ operator is contractive, and hence Algorithm \ref{alg:RL_ER_Games} can achieve the NE.

\begin{lemma} \label{lem:contract}
    For a fixed $t \in \{0,\ldots,T-1\}$, the operator $\Psi_t(\cdot|S_{[t+1:T-1]})$ is contractive for a fixed set of controllers $S_{[t+1:T-1]}$, if $\tau > 2 \gamma^2_B \gamma_{P,t} (N-1)$ where $\gamma_{P,t} := \max_{\forall i} \lVert P^i_{t+1} \rVert_F$.


\end{lemma}
\begin{proof}
    Consider $S_t,S^\prime_t \in \Ss$, where $S_t = (K^i_t, \Sigma^i_t)_{\forall i}$ and $S^\prime_t = (K^{i\prime}_t, \Sigma^{i\prime}_t)_{\forall i}$. Then,
    \begin{align*}\\[-2em]
        & \lVert \Psi_t(S_t|S_{[t+1:T-1]}) - \Psi_t(S'_t|S_{[t+1:T-1]}) \rVert_\square\\
        & =  \sum_{i \in [N]} \bigg\lVert \bigg( \frac{\tau}{2}I + R^i_t + (B^i_t)^\top P^i_{t+1} B^i_t  \bigg)^{-1} \\
        & \hspace{3.7cm} (B^i_t)^\top P^i_{t+1}   \sum_{j \neq i} B^j_t (K^j_t  - K^{j\prime}_t) \bigg\rVert_F \\
        & \leq \frac{2}{\tau} \gamma^2_B \gamma_P (N-1) \sum_{i \in [N]} \lVert K^j_t  - K^{j\prime}_t \rVert_F + \lVert \Sigma^j_t  - \Sigma^{j\prime}_t \rVert_F \\
        & = \frac{2}{\tau} \gamma^2_B \gamma_P (N-1) \lVert S_t - S^\prime_t \rVert_\square,\\[-2em]
    \end{align*}
    where $\gamma_{P,t} := \max_{\forall i} \lVert P^i_{t+1} \rVert_F$. Hence, if $\tau > 2 \gamma^2_B \gamma_P (N-1)$ then $\Psi_t$ is contractive. Thus,  we complete the proof.
\end{proof}
Utilizing this lemma and results in reference \cite{zaman2024independent}
, we now prove that if the number of inner-loop iterations $L$ in Algorithm \ref{alg:RL_ER_Games} is large enough, then the output of the algorithm $(S_t)_{0 \leq t \leq T-1} \approx (S^*_t)_{0 \leq t \leq T-1}$, where $(S^*_t)_{0 \leq t \leq T-1}$ is the NE of the ERGS LQ Game \eqref{eq:dyn}-\eqref{eq:ent_reg_cost}. Furthermore, the Nash error decreases at a linear rate with respect to the inner-loop iterations $L$.
\begin{theorem} \label{thm:main_res}
    The output of Algorithm \ref{alg:RL_ER_Games}, $(S_t)_{0 \leq t \leq T-1}$ is $\epsilon > 0$ close to the NE $(S^*_t)_{0 \leq t \leq T-1}$, s.t. $\sum_{0 \leq t \leq T-1} \lVert S_t - S^*_t \rVert_\square = \Os(\epsilon)$, if $L = \Os(\log(1/\epsilon))$ and $\tau > 2 \gamma^2_B (\gamma^*_P + c \epsilon) (N-1)$, where $c$ is a model dependent constant.
\end{theorem}
\begin{proof}
    We prove this result by first defining an auxiliary game. Consider a game with dynamics given by \eqref{eq:dyn} but with cost function 
    \begin{align}
    & J^i (x,\pi^i,\pi^{-i}) = \EE \bigg[ \sum_{t=0}^{T-1} x^\top_t Q^i_t x_t \label{eq:ent_reg_cost_equiv}  \\
    & \hspace{1cm}  + (u^i_t)^\top \big((\nicefrac{\tau}{2}) I + R^i_t \big) u^i_t + x^\top_T Q^i_T x_T \bigg| x_0 = x \nonumber \bigg].
\end{align}
From results in \cite{zaman2024independent, bacsar1998dynamic}, we know that the NE controller for this game would be linear $u^{i*}_t = K^{i*}_t x$ and deterministic, where $K^{i*}_t$ is defined in \eqref{eq:Nash_pol} and the adjoining coupled Riccati equations are given by \eqref{eq:CARE}. Hence, this game can be thought of as the deterministic version of the ERGS game \eqref{eq:dyn}-\eqref{eq:ent_reg_cost}.

Having formulated the auxiliary game, which is a LQ game, we can utilize  \cite[Theorem 4.5]{zaman2024independent}, which relies on the \textit{good} convergence of the inner loop (lines 3-6 in Algorithm \ref{alg:RL_ER_Games}). This good convergence of the inner-loop is guaranteed under $\tau > 2 \gamma^2_B (\gamma^*_P + c \epsilon) (N-1)$ since \cite[Theorem 4.5]{zaman2024independent} inductively proves that if at time $t+1$, $\lVert P^i_{t+1} - P^{i*}_{t+1} \rVert = \Os(\epsilon)$ and the inner loop is contractive (which is true under  $\tau > 2 \gamma^2_B (\gamma^*_P + c \epsilon) (N-1)$), then $\lVert P^i_t - P^{i*}_t \rVert = \Os(\epsilon)$. Since $\lVert P^i_t - P^{i*}_t \rVert = \Os(\epsilon)$ and if $\epsilon >0$ is small enough, $\lVert \Sigma^i_t - \Sigma^{i*}_t \rVert = \Os(\epsilon)$ where $\Sigma^{i}_t = \big(I + 2(R^i_t + (B^i_t)^\top P^i_{t+1} B^i_t)/\tau \big)^{-1}$. The proof is thus complete. 
\end{proof}
Notice that the condition in Theorem \ref{thm:main_res} is slightly stronger than Assumption \ref{asm:tau} and ensures linear convergence of Algorithm \ref{alg:RL_ER_Games} to the NE of the ERGS. Next we propose an augmentation technique to achieve an $\epsilon$-NE of the ERGS if Assumption \ref{asm:tau} is not satisfied.

\subsection{$\delta$-augmented Entropy-Regularized Game} \label{subsec:delta_aug}
Now we investigate the case where the entropy regularization parameter $\tau$ does not satisfy Assumption \ref{asm:tau} leading to possible non-uniqueness and indeterminacy of NE. Let us define an augmented game with an augmented entropy regularized parameter $\utau := \tau + \delta$, dynamics \eqref{eq:dyn} and cost function
\begin{align}
    & \uJ^i (x,S^i,S^{-i}) := \EE \bigg[ \sum_{t=0}^{T-1} x^\top_t Q^i_t x_t + (u^i_t)^\top R^i_t u^i_t  \label{eq:ent_reg_cost_aug} \\
    & \hspace{2cm} + \utau \log \bigg( \frac{\pi^i_t(u^i_t | x_t)}{\mu^i_t(u^i_t | x_t)} \bigg) + x^\top_T Q^i_T x_T \bigg| x_0 = x \bigg]. \nonumber
\end{align}
where $\delta \in \RR^+$ is chosen such that $\utau > 2 \gamma^2_B \ugamma^*_P (N-1)$ where $\ugamma^*_P := \max_{\forall i,t} \lVert \uP^{i*}_t \rVert_F + c\epsilon$ ($c$ and $\epsilon$ similar to Theorem \ref{thm:main_res}) and $\uP^{i*}_t$ are solutions to the coupled Riccati equations of the augmented problem. Using results in the previous sections the NE of the augmented game exists, is unique and can be found using Algorithm \ref{alg:RL_ER_Games}. Furthermore, as shown below the NE of the augmented game will be $\Os(\delta)$-NE of the original game. We denote the unique NE of the augmented game by $(\uS^{i*}_t)_{\forall i,t} = (\uK^{i*}_t,\uSigma^{i*}_t)_{\forall i,t}$. 
\begin{theorem}
The NE $(\uS^{i*}_t)_{\forall i,t} = (\uK^{i*}_t,\uSigma^{i*}_t)_{\forall i,t}$ of the augmented game is $\Os(\delta)$, i.e.,
    \begin{align*}
        J^i(x,\uS^{i*},\uS^{-i*}) \! \leq \! J^i(x,S^i,\uS^{-i*}) \! + \! \Os(\delta), \forall i \in [N], S^i \in \Ss^i. 
    \end{align*}
\end{theorem}
\begin{proof}
Recalling the value $J^i_t$ with entropy regularization parameter $\tau$ in \eqref{eq:ent_reg_cost}, we first prove that under the set of linear Gaussian controllers $(S^i_t)_{\forall i,t} = (K^i_t,\Sigma^i_t)_{\forall i,t}$ the cost function in \eqref{eq:ent_reg_cost} has a quadratic structure, i.e., 
        $J^i_t(x,S^i,S^{-i}) = x^\top P^i_t x + q^i_t$. 
    The base case is true since $J^i_T(x,S^i,S^{-i}) = x^\top Q^i_T x$. Let us assume $J^i_{t+1}(x,S^i,S^{-i}) = x^\top P^i_{t+1} x + q^i_{t+1}$. Then, unrolling $J^i_t$ by one step, we get
    \begin{align}
        & J^i_t(x,S^i,S^{-i}) = x^\top Q^i_t x + \EE \Big[(u^i_t)^\top R^i_t u^i_t + \tau \log \frac{\pi^i_t(u^i_t|x)}{\pi^i_t(u^i_t|x)} \nonumber \\
        & (A_t x + \sum_{j=1}^N B^j_t u^j_t + \omega_t)^\top\! P^i_{t+1} (A_t x + \sum_{j=1}^N B^j_t u^j_t + \omega_t) \Big] \hspace{-0.1cm} +\! q^i_{t+1} \nonumber \\
        & = x^\top Q^i_t x + (K^i_t x)^\top R^i_t (K^i_t x) + \tr(\Sigma^i_t R^i_t) + q^i_{t+1} \nonumber \\
        & + \frac{\tau}{2} \big(\tr(\Sigma^i_t) + (K^i_t x)^\top (K^i_t x) - p - \log(|\Sigma^i_t|)\big) \nonumber  \\
        & + x^\top \Big( A_t + \sum_{j=1}^N B^j_t K^j_t \Big)^\top P^i_{t+1} \Big( A_t + \sum_{j=1}^N B^j_t K^j_t \Big) x \nonumber \\
        & + \tr\bigg( (\Sigma + \sum_{j=1}^N \Sigma^j_t)P^i_{t+1} \bigg) = x^\top P^i_t x + q^i_t, \label{eq:value_ent_reg}
    \end{align}
    where matrix $P^i_t$ and scalar $q^i_t$ can be written recursively as
    \begin{align}
        P^i_t & = Q^i_t + (K^i_t)^\top ((\nicefrac{\tau}{2}) I +  R^i_t) K^i_t \nonumber \\
        & + \Big( A_t + \sum_{j=1}^N B^j_t K^j_t \Big)^\top\!\! P^i_{t+1} \Big( A_t + \sum_{j=1}^N B^j_t K^j_t \Big), \hspace{0.1cm} P^i_T = Q^i_T \nonumber \\
        q^i_t & = q^i_{t+1} + \tr(\Sigma^i_t ((\nicefrac{\tau}{2}) I + R^i_t)) - \frac{\tau}{2} \big(p + \log(|\Sigma^i_t|)\big) \nonumber \\
        & + \tr\bigg( (\Sigma + \sum_{j=1}^N \Sigma^j_t)P^i_{t+1} \bigg), \hspace{0.1cm} q^i_T = 0. \label{eq:P_N}
    \end{align}
    Hence, we have proven the quadratic form of $J^i_t$. The cost function for the $\delta$--augmented game is given by  
    \begin{align*}
        & \uJ^i_{t} (x,S^i,S^{-i}) = \EE \Big[ \sum_{s=t}^{T-1} x^\top_s\big( Q^i_s \\
        & \hspace{1.2cm} + (K^i_s)^\top ((\nicefrac{\utau}{2}) I + R^i_s) K^i_s \big) x_s + x^\top_T Q^i_T x_T\Big| x_t = x \Big],
    \end{align*}
    for all $i \in [N]$. The value function for this game has a structure similar to \eqref{eq:value_ent_reg},
    \begin{align*}
        \uJ^i_{t}(x_t,K^i,K^{-i}) = x^\top_t \uP^i_t x_t + \uq^i_t,
    \end{align*}
    where $\uP^i_t$ is defined recursively as
    \begin{align}
        \uP^i_t & =  Q^i_t + (K^i_t)^\top (\nicefrac{\utau}{2}I + R^i_t) K^i_t \nonumber\\
        & + \Big(A_t + \sum_{j=1}^N B^j_t K^j_t \Big)^\top \!\!\uP^i_{t+1} \Big(A_t + \sum_{j=1}^N B^j_t K^j_t \Big), \nonumber \end{align}
        and $\uP^i_T = Q^i_T$. Similarly, $\uq^i_t$ is defined recursively by
        \begin{align}
        \uq^i_t  =& \uq^i_{t+1} + \tr(\Sigma^i_t (\nicefrac{\utau}{2} I + R^i_t)) - \frac{\utau}{2} \big(p + \log(|\Sigma^i_t|)\big) \nonumber \\
        & + \tr\bigg( (\Sigma + \sum_{j=1}^N \Sigma^j_t)\uP^i_{t+1} \bigg), \uq^i_T = 0. \label{eq:uP_uN}
    \end{align}
    Next, we find the difference between the two costs as follows:
    \begin{align}
        & \lvert \uJ^i_{t}(x,S^i,S^{-i}) - J^i_{t}(x,S^i,S^{-i}) \rvert \label{eq:uJ-J} \\
        & \hspace{3cm} = | x^\top \big( \uP^i_t - P^i_t \big) x  + \uq^i_t - q^i_t | \nonumber 
    \end{align}
    Thus, to quantify the difference $\uJ^i_{t}(K^i,K^{-i}) - J^i_{t}(K^i,K^{-i})$, we need to upper bound 
    $\lVert \uP^i_t - P^i_t \rVert_F$ and $\lvert \uq^i_t - q^i_t \rvert$. Thus, using \eqref{eq:uP_uN} and \eqref{eq:P_N}, we can deduce that 
    \begin{align}
        \lvert \uq^i_t - q^i_t \rvert = \Os \Big(\max_{t < s \leq T} \lVert \uP^i_s - P^i_s \rVert_F \Big). \label{eq:uN-N}
    \end{align}
    In the interest of conciseness, we omit these computations. Now, we bound the quantity $\lVert \uP^i_t - P^i_t \rVert_F$ as follows:
    \begin{align}
        & \lVert \uP^i_t - P^i_t \rVert_F = \Big\lVert  \delta/2(K^i_t)^\top K^i_t + \nonumber \\
        & \hspace{0.4cm} \Big( A_t + \sum_{j=1}^N B^j_t K^j_t \Big)^\top (\uP^i_{t+1} - P^i_{t+1}) \Big( A_t + \sum_{j=1}^N B^j_t K^j_t \Big) \Big\rVert_F \nonumber \\
        & \leq \delta/2 \lVert K^i_t \rVert^2_F + \Big(\gamma_A + \sum_{j=1}^N \gamma_B \lVert K^j_t \rVert_F \Big)^2 \lVert \uP^i_{t+1} - P^i_{t+1} \rVert_F \nonumber \\
        & \leq \delta \underbrace{(D^2/2)}_{c_1} + \underbrace{(\gamma_A + N \gamma_B D )^2}_{c_2} \lVert \uP^i_{t+1} - P^i_{t+1} \rVert_F, \label{eq:uP-P}
    \end{align}
    where $\gamma_A = \max_{\forall t} \lVert A_t \rVert_F$. The last inequality utilizes the fact that since $\utau > 2 \gamma^2_B \ugamma^*_P (N-1)$, $\Psi_t$ operator for this game will be contractive, and hence the norm of iterates $K^i_t$ will be bounded. Using \eqref{eq:uP-P}, we can recursively say that
    \begin{align*}
        \lVert \uP^i_t - P^i_t \rVert_F & \leq c_1 \delta/2 \sum_{s=t}^{T-1} c^{s-t}_2  + c^{T-t}_2 \lVert \uP^i_T - P^i_T \rVert_F = \Os(\delta) 
    \end{align*}
    Similarly, we have
    \begin{align}
        \lvert \uq^i_t - q^i_t \rvert = \Os(\delta). \label{eq:uN-N_Os}
    \end{align}
    Thus, by using \eqref{eq:uJ-J} and \eqref{eq:uN-N} with $t=0$, we obtain 
    \begin{align}
        \lvert \uJ^i(x,S^i,S^{-i} ) - J^i(x,S^i,S^{-i}) \rvert = \Os(\delta). \label{eq:cost_diff_aug}
    \end{align}
    Further, using \eqref{eq:cost_diff_aug} and denoting $(\uS^{i*}_t)_{\forall i,t} = (\uK^{i*}_t,\uSigma^{i*}_t)_{\forall i,t}$ as the unique NE of the $\delta$--augmented game with entropy regularization parameter $\utau$ (which is bound to exist due to the condition on $\utau$), for any $S^i = (K^i,\Sigma^i) \in \Ss^i$, we have
    \begin{align}
        & J^i  (x,\uS^{i*},\uS^{-i*}) - \uJ^i  (x,S^i, \uS^{-i*} ) \nonumber \\
        & \hspace{0.4cm} \leq J^i  (x,\uS^{i*},\uS^{-i*}) - \uJ^i  (x,\uS^{i*}, \uS^{-i*} ) = \Os(\delta). \label{eq:eps_NE_aug_inter}
    \end{align}
    Next, let us fix $i\! \!\in\!\![N]$ and define a set of controllers $(S^i,\uS^{-i*})$ where $S^i \in \Ss^i$. We know from our earlier analysis that $\uJ^i(x,S^i,\uS^{-i*}) = x^\top \tilde\uP^i_t x + \tilde \uq^i_t$ and $J^i(x,S^i,\uS^{-i*}) = x^\top \tP^i_t x + \tq^i_t$, where $\tP^i_t,$ $ \tilde{\uP}^i_t,$ $ \tq^i_t$, and $\tilde{\uq}^i_t$ are defined recursively as
    \begin{align*}
        & \tilde\uP^i_t = Q^i_t + (K^i_t)^\top ((\nicefrac{\utau}{2})I + R^i_t) K^i_t + \nonumber\\
        & (A_t \!+\! B^i_t K^i_t + \!\sum_{j \neq i}^N B^j_t \uK^{j*}_t)^\top \tilde\uP^i_{t+1} (A_t \!+\! B^i_t K^i_t + \!\sum_{j \neq i}^N B^j_t \uK^{j*}_t),  \nonumber \\
        & \tilde\uq^i_t = \tilde\uq^i_{t+1} + \tr(\Sigma^i_t ((\nicefrac{\utau}{2}) I + R^i_t)) - \frac{\utau}{2} \big(p + \log(|\Sigma^i_t|)\big) \nonumber \\
        & + \tr\bigg( (\Sigma + \Sigma^i_t + \sum_{j \neq i} \uSigma^{j*}_t)\tilde\uP^i_{t+1} \bigg), \hspace{0.1cm} \tilde\uq^i_T = 0, \\
        & \tilde P^i_t = Q^i_t + (K^i_t)^\top ((\nicefrac{\tau}{2}) I +  R^i_t) K^i_t + \\
        & (A_t \!+\! B^i_t K^i_t + \!\sum_{j \neq i}^N B^j_t \uK^{j*}_t)^\top \tilde P^i_{t+1} (A_t \!+\! B^i_t K^i_t +\! \sum_{j \neq i}^N B^j_t \uK^{j*}_t), \nonumber \end{align*}\begin{align*}
        & \tilde q^i_t = \tilde q^i_{t+1} + \tr(\Sigma^i_t ((\nicefrac{\tau}{2}) I + R^i_t)) - \frac{\tau}{2} \big(p + \log(|\Sigma^i_t|)\big) \nonumber \\
        & + \tr\bigg( (\Sigma + \Sigma^i_t + \sum_{j \neq i} \uSigma^{j*}_t)\tP^i_{t+1} \bigg), \hspace{0.1cm} \tq^i_T = 0.\\[-2.2em]
    \end{align*}
    where $\tilde\uP^i_T = \tilde P^i_T= Q^i_T$. Using these expressions, we can deduce that
    \begin{align}
        &J^i (x,S^i,\uS^{-i*}) \nonumber \\
        & = \uJ^i (x,S^i,\uS^{-i*} ) - x^\top (\tilde\uP^i_0 - \tP^i_0 )x - \tilde\uq^i_0 + \tq^i_0 \nonumber \\
        & \geq \uJ^i (x,\uS^{i*},\uS^{-i*} ) - x^\top (\tilde\uP^i_0 - \tP^i_0 )x - \tilde\uq^i_0 + \tq^i_0  \label{eq:eps_NE_aug_inter_2}
    \end{align}
    Using analysis similar to \eqref{eq:uP_uN}-\eqref{eq:uN-N_Os}, we can further conclude that $\lVert \tilde\uP^i_0 - \tP^i_0 \rVert_F = \Os(\delta)$ and $\lvert \tilde\uq^i_0 - \tq^i_0 \rvert = \Os(\delta)$. Hence, $ | \uJ^i(x,S^i,\uS^{-i*}) - J^i(x,S^i,\uS^{-i*}) | = \Os(\delta)$. Finally, using \eqref{eq:eps_NE_aug_inter}-\eqref{eq:eps_NE_aug_inter_2}, we can write
    \begin{align*}
        & | J^i  (x,\uS^{i*},\uS^{-i*}) - J^i (x,S^i,\uS^{-i*}) | \\
        & = | J^i  (x,\uS^{i*},\uS^{-i*}) - \uJ^i  (x,\uS^i, \uS^{-i*} ) \\
        & \hspace{0.5cm} + \uJ^i (x,\uS^{i*},\uS^{-i*} ) - J^i (x,S^i,\uS^{-i*}) | \\
        & \leq | J^i  (x,\uS^{i*},\uS^{-i*}) - \uJ^i  (x,\uS^i, \uS^{-i*} ) | + | J^i (x,S^i,\uS^{-i*}) \\
        & \hspace{0.5cm}  - J^i (x,S^i,\uS^{-i*}) | + | x^\top (\tilde\uP^i_0 - \tP^i_0 ) x | + | \tilde\uq^i_0 - \tq^i_0 | \\
        & = \big| J^i  (x,\uS^{i*},\uS^{-i*}) - \uJ^i  (x,\uS^i, \uS^{-i*} ) \big| \\
        & \hspace{0.5cm} + \big| x^\top (\tilde\uP^i_0 - \tP^i_0 )x \big| + \big| \tilde\uq^i_0 - \tq^i_0 \big| = \Os(\delta)
    \end{align*}
    which completes the proof.
\end{proof}

\section{Conclusion \& Future Work}
In this work, we have formulated the relative entropy-regularized linear-quadratic general-sum $N$-agent games and proved that their Nash Equilibria (NE) lie within the domain of linear Gaussian policies. Furthermore, we have delineated sufficient conditions for the uniqueness of the same. Subsequently, we have presented a Policy Optimization (PO) algorithm capable of verifiably attaining the NE at a linear convergence rate under suitable conditions on the regularization parameter. Finally, we have introduced a $\delta$-augmentation technique to obtain an $\epsilon$-NE of the game in instances where the regularization proves inadequate. Future directions include extending the PO algorithm to a data-driven 
RL framework that enables each agent to autonomously learn the NE. Additionally, we intend to explore the ramifications of alternative regularization techniques, such as cross-entropy regularization and Hellinger-Bhattacharya Regularization, on both the NE and the efficacy of PO algorithms in achieving convergence.

\bibliography{references}
\bibliographystyle{IEEEtran} 
\end{document}